\documentclass[runningheads]{llncs}
\usepackage{graphicx}
\usepackage{color}
\usepackage{hyperref}
\usepackage{float}

\usepackage{bm}

\usepackage[X2,T2A]{fontenc}

\newcommand{\bx}{\mathbf{x}}
\newcommand{\by}{\mathbf{y}}
\newcommand{\Bf}{\mathbf{f}}
\newcommand{\bg}{\mathbf{g}}
\newcommand{\bh}{\mathbf{h}}
\newcommand{\bp}{\mathbf{p}}
\newcommand{\bq}{\mathbf{q}}
\newcommand{\bz}{\mathbf{z}}
\newcommand{\br}{\mathbf{r}}

\usepackage[ruled,vlined,linesnumbered]{algorithm2e}
\usepackage{enumitem}

\usepackage{amsmath,amssymb}
\usepackage{multirow} % Add this line to your preamble
\usepackage{tikz}
\usepackage{forest}
\usepackage{comment}
\usepackage{graphicx}
\usepackage{subcaption}

\begin{document}

\title{Dissipative quadratizations of\\ polynomial ODE systems
\thanks{This work has partly been supported by the French ANR-22-CE48-0008 OCCAM and ANR-22-CE48-0016 NODE projects. We would like to thank Andrey Bychkov, Marcelo Forets, Christian Schilling, Boris Kramer for helpful discussions and the referees for careful reading and detailed comments.}
}
\author{Yubo Cai\inst{1}\orcidID{0009-0005-9278-6315}
\and
Gleb Pogudin\inst{2}\orcidID{0000-0002-5731-8242}
}
\authorrunning{Y. Cai, G. Pogudin}

\institute{
\'Ecole Polytechnique, Institute Polytechnique de Paris, Palaiseau, France
\email{yubo.cai@polytechnique.edu} \and
LIX, CNRS, \'Ecole Polytechnique, Institute Polytechnique de Paris, Palaiseau, France \email{gleb.pogudin@polytechnique.edu}
}

\maketitle              % typeset the header of the contribution
\begin{abstract}
Quadratization refers to a transformation of an arbitrary system of polynomial ordinary differential equations to a system with at most quadratic right-hand side.  
Such a transformation unveils new variables and model structures that facilitate model analysis, simulation, and control and offer a convenient parameterization for data-driven approaches. 
Quadratization techniques have found applications in diverse fields, including systems theory, fluid mechanics, chemical reaction modeling, and mathematical analysis.

In this study, we focus on quadratizations that preserve the stability properties of the original model, specifically dissipativity at given equilibria. 
This preservation is desirable in many applications of quadratization including reachability analysis and synthetic biology. 
We establish the existence of dissipativity-preserving quadratizations, develop an algorithm for their computation, and demonstrate it in several case studies.

\end{abstract}
%
%
%

%%%%%%%%%%%%%%%%%%%%%%%%%%%%%%%%%%%%

\section{Introduction}
Systems of ordinary differential equations (ODEs) are the standard choice when it comes to modeling processes happening in continuous time, for example, 
in the sciences and engineering. 
For a given dynamical process, one can derive different ODE models, in particular, by choosing different sets of variables.
It has been observed in a variety of areas and contexts that these choices may have a significant impact on the utility and relevance of the resulting model, and a number of different types of variable transformations have been studied.

In this paper, we will study one such transformation, \emph{quadratization}, which aims at transforming an ODE system to a system where the right-hand side consists of polynomials of degree at most two.
Let us illustrate this transformation on a toy example: we start with a scalar ODE $x' = x^3$ in a single variable $x = x(t)$ with cubic right-hand side.
If we now augment the state space with an additional coordinate $y = x^2$, we can write the original equation as $x' = xy$ with quadratic right-hand side, and we can do the same for $y'$:
\[
  y' = 2xx' = 2x^4 = 2y^2.
\]
So, the transformation in this case is the following:
\[
x' = x^3 \quad\to\quad \begin{cases}
    x' = xy,\\
    y' = 2y^2.
\end{cases}
\]
It turns out that every polynomial ODE system can be similarly lifted to an at most quadratic one: this fact has been established at least 100 years ago~\cite{Appelroth1902,Lagutinskii} and has been rediscovered several times since then~\cite{CPSW05,carravetta2015global,carravetta2020solution,gu2011qlmor,kerner1981universal}.
In the recent years quadratization has been used in a number of application areas including model order reduction~\cite{bennerBreiten2015twoSided,BOPK23,gu2011qlmor,KW18nonlinearMORliftingPOD,KW2019_balanced_truncation_lifted_QB}, synthetic biology~\cite{fages2017strong,Hemery2022,lifeware1}, numerical integration~\cite{GD19,guillot2019generic,GCV2019}, and reachability analysis~\cite{Forets_2021}.
While it has been shown in~\cite{lifeware1} that the problem of finding the minimal number of extra variables necessary for quadratization is NP-hard, at least two practically useful software packages have been developed for performing quadratization: {\sc BioCham}~\cite{lifeware1} and {\sc QBee}~\cite{bychkov2021optimal}.

In the majority of the applications mentioned above, the constructed quadratic ODE model is further used in the context of \emph{numerical} simulations. 
It is, therefore, a natural question whether one can not only guarantee that the transformed model is at most quadratic, but also that it preserves some desirable dynamical/numerical properties of the original ODE system.
To the best of our knowledge, this question has not been studied systematically, and in this paper, we initiate this line of research by studying \emph{dissipativity-preserving quadratizations}.

We will say that an ODE system is \emph{dissipative} at an equilibrium point if the real parts of the eigenvalues of the linearization of the system around this point are negative.
In particular, dissipativity implies that the system is \emph{asymptotically stable} at this point~\cite[Theorem 8.2.2]{Hubbard}.
The main contribution of the paper is two-fold. 
First, we prove that, for every polynomial system dissipative at several equilibrium points, there exists a quadratization which is also dissipative at all these points.
Second, we design and implement an algorithm to search automatically for such quadratization attempting to minimize the dimension.
Our algorithm is based on a combinatorial condition on the new variables which is sufficient to guarantee that the resulting quadratic model can be made dissipative as well.
This combinatorial condition can be viewed as a generalization and formalization of the artificial stabilization used in~\cite[Section~4.1]{KW2019_balanced_truncation_lifted_QB}.
We implemented the new algorithm and we illustrate it in several case studies including an application for reachability analysis (in combination with the algorithm from~\cite{Forets_2021}).
Our implementation together with the examples from this paper is available at~\cite{code}.

The rest of the paper is organized as follows.
In Section~\ref{sec:preliminaries}, we introduce the main notions, quadratization, and dissipativity, and show that quadratization performed straightforwardly may not preserve dissipativity (and, thus, render the model into a numerically unstable one).
Section~\ref{sec:theory} contains the statement and the proof of the main theoretical result of the paper (Theorem~\ref{thm:main}) that there always exists a dissipativity-preserving quadratization for any collection of dissipative equilibria.
Based on the ideas from the proof, we give an algorithm (Algorithm~\ref{alg:multi_dissipativity}) for constructing such a quadratization in Section~\ref{sec:alg}.
We showcase our implementation of this algorithm on several case studies in Section~\ref{sec:case_studies}.
Concluding remarks are contained in Section~\ref{sec:conclusions}.

%%%%%%%%%%%%%%%%%%%%%%%%%%%%%%%%%%%%

\section{Preliminaries}
\label{sec:preliminaries}

Throughout this section, we will consider a polynomial ODE system, that is, a system of differential equations
\begin{equation}\label{eq:mainode}
\bx' = \bp(\bx),
\end{equation}
where $\bx = \bx(t) = (x_1(t), \ldots, x_n(t))$ is a vector of unknown functions and $\bp = (p_1, \ldots, p_n)$ is a vector of $n$-variate polynomials $p_1, \ldots, p_n \in \mathbb{R}[\bx]$.

\begin{definition}[Quadratization]\label{def:quadr}
    For a system~\eqref{eq:mainode}, \emph{quadratization} is a pair consisting of
    \begin{itemize}
        \item a list of new variables
        \[
          y_1 = g_1(\bx), \ldots, y_m = g_m(\bx)
        \]
        \item and two lists 
        \[
        \bq_1(\bx, \by) = (q_{1, 1}(\bx), \ldots, q_{1, n}(\by)) \quad\text{ and }\quad \bq_2(\bx, \by) = (q_{2, 1}(\bx, \by), \ldots, q_{2, m}(\bx, \by))
        \]
        of $m + n$-variate polynomials in $\bx$ and $\by = (y_1, \ldots, y_m)$
    \end{itemize}
    such that  the degree of each of of $\bq_1$ and $\bq_2$ is at most two and 
    \begin{equation}\label{eq:resulting_quadr}
      \bx' = \bq_1(\bx, \by) \quad \text{ and } \quad \by' = \bq_2(\bx, \by).
    \end{equation}
    If all the polynomials $g_1, \ldots, g_m$ are monomials, the quadratization is called \emph{monomial quadratization}.
\end{definition}

Note that unlike, for example, \cite[Definition~1]{bychkov2021optimal}, by \emph{quadratization} we mean not just the set of new variables but also the quadratic ODE system~\eqref{eq:resulting_quadr}.
The reason for this is that, for a fixed set of new variables, there may be many different systems of the shape (see Example~\ref{ex:quadratization}) prescribed by~\eqref{eq:resulting_quadr} exhibiting different numerical behaviors (see Example~\ref{ex:stability}).

\begin{example}[Quadratization] \label{ex:quadratization}
    Consider the following scalar ODE
    \[
    x' = -x + x^3.
    \]
    Here we have $n = 1$ and $p_1(x) = -x + x^3$.
    Consider $y = g_1(x) = x^2$. 
    Then we can write
    \begin{align*}
    x' &= -x + x^3 = -x + xy,\\
    y' &= 2xx' = -2x^2 + 2x^4 = -2y + 2y^2.
    \end{align*}
    Therefore, one possible quadratization is given by
    \[
     g_1(x) = x^2, \quad q_{1, 1}(x, y) = -x + xy, \quad q_{2, 1}(x, y) = -2y + 2y^2.
    \]
    As we have mentioned above, there may be different $\bq$'s corresponding to the same $\bg$.
    In this example, we could take, for example, $q_{2, 1} = -2y + 2y^2 + 2(y - x^2) = -2x^2 + 2y^2$ or  $q_{2, 1} = y - 3x^2 + 2y^2$.
    As we will see in Example~\ref{ex:stability}, such choices may have a dramatic impact on the numerical properties of the resulting ODE system.
\end{example}

\begin{definition}[Equilibrium]
    For a polynomial ODE system~\eqref{eq:mainode}, a point $\bx^\ast \in \mathbb{R}^n$ is called an \emph{equilibrium} if $\bp(\bx^\ast) = 0$.
\end{definition}

\begin{definition}[Dissipativity]
    An ODE system~\eqref{eq:mainode} is called \emph{dissipative} at an equilibrium point $\bx^\ast$ if all the eigenvalues of the Jacobian $J(\bp)|_{\bx = \bx^\ast}$ of $\bp$ and $\bx^\ast$ have negative real part.

    It is known that a system which is dissipative at an equilibrium point $\bx^\ast$ is~\emph{asymptotically stable} at $\bx^\ast$~\cite[Theorem 8.2.2]{Hubbard}, that is, any trajectory starting in a small enough neighborhood of $\bx^\ast$ will converge to $\bx^\ast$ exponentially fast.

    Note that if $\bx^\ast = 0$, then the Jacobian at this point is simply the matrix of the linear part of $\bp(\bx)$.
\end{definition}

Assume that $\bx^\ast \in \mathbb{R}^n$ is an equilibrium of $\bx' = \bp(\bx)$, and consider a quadratization of this system as in Definition~\ref{def:quadr}.
Then a direct computation shows that $(\bx^\ast, \bg(\bx^\ast))$ is an equilibrium point of the resulting quadratic system~\eqref{eq:resulting_quadr}.

\begin{definition}[Dissipative quadratization]
    Assume that a system~\eqref{eq:mainode} is dissipative at an equilibrium point $\bx^\ast \in \mathbb{R}^n$.
    Then a quadratization given by $\bg, \bq_1$ and $\bq_2$ (see Definition~\ref{def:quadr}) is called~\emph{dissipative at~$\bx^\ast$} if the system
    \[
      \bx' = \bq_1(\bx, \by), \quad \by' = \bq_2(\bx, \by)
    \]
    is dissipative at a point $(\bx^\ast, \bg(\bx^\ast))$.
\end{definition}

The following example shows that, even for the same new variables $\by = \bg(\bx)$, different quadratizations may have significantly different stability properties. 

\begin{example}[Stable and unstable quadratizations]\label{ex:stability} 
Consider the scalar ODE $x' = -x + x^3$ from Example \ref{ex:quadratization}. 
We have already found a quadratization for it using a new variable $y = g_1(x) = x^2$ with the resulting quadratic system being
\begin{equation}\label{eq:stable}
    x' = -x + xy\quad \text{ and } y' = -2y + 2y^2.
\end{equation}
We notice that we can add/subtract $y - x^2$ with any coefficients from the right-hand sides of the system. 
For example, we can obtain:
\begin{equation}\label{eq:unstable}
    x' = -x + xy\quad \text{and} \quad y' = -2y + 2y^2 + 12(y - x^2) = 10y - 12x^2 + 2y^2.
\end{equation}
Both systems above are quadratizations of the original system and, thus, mathematically, for any initial condition $(x_0, y_0)$ satisfying $y_0^2 = x_0^2$, they must follow the same trajectory.
However, \eqref{eq:stable} is stable at $(0, 0)$ while~\eqref{eq:unstable} is not.
By numerically integrating them, we can observe in Figure~\ref{fig:stability} that in practice \eqref{eq:stable} reflects the dynamics of the original equation accurately and~\eqref{eq:unstable} heavily suffers from numerical instability.

\begin{figure}[h!]
    \centering
    \includegraphics[width=1.0\textwidth]{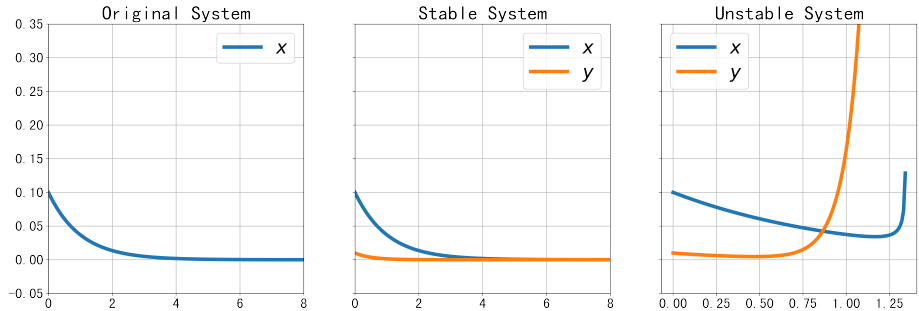}
    \caption{Plot of the original equation, \eqref{eq:stable}, and~\eqref{eq:unstable} with initial condition $\mathcal{X}_0 = [x_{0}, y_{0}=x_{0}^{2}]=[0.1, 0.01]$. Numerical method: ``LSODA'' (uses hybrid Adams/BDF method with automatic stiffness detection) in scipy.integrate.solve\_ivp package \cite{hindmarsh1983odepack,petzold1983automatic}.}
    \label{fig:stability}
\end{figure}
\end{example}

%%%%%%%%%%%%%%%%%%%%%%%%%%%%%%%%%%%%

\section{Existence of dissipativity-preserving quadratizations}\label{sec:theory}

The main result of this section is the following theorem.
Its proof is constructive and is used to design an algorithm in Section~\ref{sec:alg}.

\begin{theorem}\label{thm:main}
    For every polynomial ODE system $\bx' = \bp(\bx)$, there exists a quadratization that is dissipative at all the dissipative equilibria of $\bx' = \bp(\bx)$.
\end{theorem}

\begin{remark}
    In fact, the key ingredients of the proof, Propositions~\ref{prop:inner_exists} and~\ref{prop:gauging}, imply a stronger statement: for every set of finitely many equilibria, there is a quadratization such that the number of nonnegative eigenvalues of the Jacobian of the quadratic system at these points is the same as for the original system.
\end{remark}

The rest of the section will be devoted to proving Theorem~\ref{thm:main}. 
The main technical notion will be an \emph{inner-quadratic} set of polynomials.
\begin{definition}[Inner-quadratic set]\label{def:inner_quadr}
As finite set $g_1(\bx), \ldots, g_m(\bx) \in \mathbb{R}[\bx]$ of nonconstant polynomials in $\bx = (x_1, \ldots, x_n)$ is called \emph{inner-quadratic} if, for every $1 \leqslant i \leqslant m$, there exist (not necessarily distinct) $a, b \in \{x_1, \ldots, x_n, g_1, \ldots, g_m\}$ such that $g_i = ab$.

A quadratization will be called~\emph{inner-quadratic} if the set of new variables $\bg$ is inner-quadratic.
We will also always assume that the new variables are sorted by degree, that is, $\deg g_1 \leqslant \deg g_2 \leqslant \ldots \leqslant \deg g_m$.
\end{definition}

The rationale behind the notion of inner-quadratic quadratization is that at most quadratic relations between the new variables give us the flexibility to ``tune'' the right-hand side of the resulting quadratic system in the same fashion as we added a multiple of $y - x^2$ in~Example~\ref{ex:stability}.
These additional terms force the trajectory to stay on the image of the map $\bx \to (\bx, \by)$, on which the properties of the original dynamics (such as dissipativity) are preserved.
The following definition formalizes this observation.

\begin{definition}[Stabilizers]\label{def:stabilizer}
    Consider a polynomial ODE system $\bx' = \bp(\bx)$ and its inner-quadratic quadratization given by $m$ new variables $\by = \bg(\bx)$ and right-hand side $\bq_1(\bx, \by), \bq_2(\bx, \by)$ of the resulting quadratic system as in Definition~\ref{def:quadr}.
    For every $1 \leqslant i \leqslant m$, by the definition of inner-quadratic set, there exist $a_i, b_i \in \{\bx, y_1, \ldots, y_{i - 1}\}$ such that the equality $y_i = a_ib_i$ holds if we replace each $y_j$ with $g_j(\bx)$.
    We define \emph{the $i$-th stabilizer} by $h_i(\bx, \by) := y_i - a_ib_i$.

    Since each stabilizer is at most quadratic and vanishes under the substitution $\by = \bg(\bx)$, adding any stabilizer to any of $\bq_1, \bq_2$ still yields a quadratization of $\bx' = \bp(\bx)$.
\end{definition}

\begin{example}[Stabilizers]\label{eg:stabilizer_example}
   Let us give an example of the stabilizers. Consider a system:
    $$
    x_1^{\prime}=-3 x_1+x_2^4 , \quad 
    x_2^{\prime}=-2 x_2+x_1^2.
    $$
    By applying \textbf{Algorithm \ref{alg:inner-quadratization}}, we introduce the following new variables to obtain an inner-quadratic quadratization:
    $$
    y_1=x_1^2, \quad
    y_2=x_2^2, \quad
    y_3=x_1 x_2, \quad
    y_4=x_2^3 = x_{2}y_{2}.
    $$
    Then the corresponding stabilizers, according to the definition above, will be:
    \begin{align*}
    h_{1}(\bx, \by) &=y_1-x_1^2, & h_{2}(\bx, \by) &=y_2-x_2^2,\\
    h_{3}(\bx, \by) &=y_3-x_1 x_2, &
    h_{4}(\bx, \by) &= y_4-x_{2}y_{2}
    \end{align*}
\end{example}

Theorem~\ref{thm:main} follows directly from the following two properties of inner-quadratic quadratizations:
\begin{itemize}
    \item every polynomial ODE system has an inner-quadratic quadratization (Proposition~\ref{prop:inner_exists});
    \item for any inner-quadratic quadratization, one can modify the right-hand sides of the quadratic system (but not the new variables) using the stabilizers in order to obtain a dissipativity-preserving quadratization (Proposition~\ref{prop:gauging}).
\end{itemize}

%%%%%%%%%%%%%%%%%%%%%%%%%%%%

\begin{proposition}\label{prop:inner_exists}
    Every polynomial ODE system $\bx' = \bp(\bx)$ admits an inner-quadratic quadratization.
    Furthermore, it can be chosen to be a monomial quadratization.
\end{proposition}

\begin{proof}
    We will show that the quadratization which is typically used to prove the existence of a quadratization for every polynomial ODE system (see, e.g.~\cite[Theorem~1]{Carothers2005}) is in fact inner-quadratic.
    For every $1 \leqslant i \leqslant n$, we introduce $d_i = \max\limits_{1 \leqslant j \leqslant n} \deg_{x_i} p_j$.
    Then it is proven in~\cite[Theorem~1]{Carothers2005}  that the following set of new variables yields a quadratization of $\bx' = \bp(\bx)$:
    \[
    \mathcal{M} = \{x_1^{i_1}\ldots x_n^{i_n} \mid \forall j\colon 0 \leqslant i_j \leqslant d_j, \; \sum i_j > 1\}.
    \]
    Let $g \in \mathcal{M}$. Then there exists $1 \leqslant j \leqslant n$ such that $\deg_{x_j} g > 0$. 
    Then we can write $g = (g / x_j) \cdot x_j$, where $g / x_j$ is either in $\mathcal{M}$ or belongs to $\{x_1, \ldots, x_n\}$.
    Thus, $\mathcal{M}$ is an inner-quadratic set.
\end{proof}

%%%%%%%%%%%%%%%%%%%%%%%%%%%%

\begin{proposition}\label{prop:gauging}
    Consider a system $\bx' = \bp(\bx)$ and its inner-quadratic quadratization defined by new variables $\bg(\bx)$ and the new right-hand side $\bq_1, \bq_2$ as in Definition~\ref{def:quadr}.
    Let $\bx_1^*, \ldots, \bx_\ell^*$ be a finite subset of the equilibria of the system.
    Then there exist vectors of quadratic polynomials $\br_1(\bx, \by), \br_2(\bx, \by)$ such that $\bg, \br_1, \br_2$ define a quadratization for which the eigenvalues of the Jacobian at each equilibrium point of the form $(\bx_i^\ast, \bg(\bx_i^*))$ are the union of the eigenvalues of $J(\bp)|_{\bx = \bx_i^*}$ and a set of numbers with negative real part.
\end{proposition}

\begin{corollary}\label{cor:gauging}
    Consider a system $\bx' = \bp(\bx)$ and its inner-quadratic quadratization defined by new variables $\bg(\bx)$ and the new right-hand side $\bq_1, \bq_2$ as in Definition~\ref{def:quadr}.
    Then there exist vectors of quadratic polynomials $\br_1(\bx, \by), \br_2(\bx, \by)$ such that $\bg, \br_1, \br_2$ define a quadratization which is dissipative at every dissipative equilibrium of $\bx'=\bp(\bx)$.
\end{corollary}

\begin{proof}[Proof of Corollary~\ref{cor:gauging}]
    Since each dissipative equilibrium of the system is an isolated root of the polynomial system obtained by equating the right-hand side to zero, there are only finitely many of them.
    So we apply Proposition~\ref{prop:gauging} to this finite set of equilibria and obtain the desired quadratization.
\end{proof}

Before proving Proposition~\ref{prop:gauging}, we establish a useful linear-algebraic lemma.

\begin{lemma}\label{lem:add_triang}
    Let $A \in \mathbb{R}^{n\times n}$ be a square matrix and $B \in \mathbb{R}^{n\times n}$ be an upper triangular matrix with ones on the diagonal.
    Then there exists $\lambda_0 \in \mathbb{R}$ such that, for every $\lambda > \lambda_0$, the real parts of all the eigenvalues of $A - \lambda B$ are negative.
\end{lemma}

\begin{proof}
    Consider the characteristic polynomial of $A - \lambda B$.
    It can be written as 
    \[
    \operatorname{det}(A - \lambda B - t I) = (-\lambda)^n \operatorname{det} (B - A/\lambda + (t/\lambda I)).
    \]
    We set $T = t/\lambda$ and rewrite the latter determinant as $Q(T, 1/\lambda) := \operatorname{det}((B + T \cdot I) - A/\lambda)$.
    Since $Q$ is the determinant of a matrix with the entries linear in $T$ and $1/\lambda$, it is a bivariate polynomial in $T$ and $1/\lambda$ of total degree at most $n$.
    Furthermore, if we set $\lambda = \infty$ (equivalently, if we set $1/\lambda = 0$), we have $Q(T, 0) = \operatorname{det}(B + T\cdot I)$.
    Since $B$ is upper-triangular with ones on the diagonal, this determinant is equal to $\operatorname{det}(B + T\cdot I) = (T + 1)^n$, so $Q(T, 1/\lambda)$ can be written as $Q(T, 1/\lambda) = (T + 1)^n + \frac{1}{\lambda} p\left(T + 1, \frac{1}{\lambda}\right)$, where $p$ is a bivariate polynomial of the total degree at most $n - 1$ in $T + 1$ and $\frac{1}{\lambda}$.
    Let $C$ be an upper bound for the absolute value of the coefficients of $p$. 
    Then, for $\lambda > 1$, we can bound:
    \[
     \left\lvert p\left(T + 1, \frac{1}{\lambda}\right)\right\rvert < C n^2 \max(|T + 1|^{n - 1}, 1). 
    \]
    Let $T_0$ be any root of $Q(T, 1/\lambda)$.
    Then we have
    \[
     |T_0 + 1|^n \leqslant \frac{1}{\lambda} C n^2 \max(|T_0 + 1|^{n - 1}, 1). 
    \]
    Let us take $\lambda > Cn^2$.
    Then
    \[
    |T_0 + 1|^n < \max(|T_0 + 1|^{n - 1}, 1) \implies |T_0 + 1| < 1.
    \]
     So, in this case, the real part of any root of $Q$ will be negative.
    Then the same is true for the characteristic polynomial of $A - \lambda B$ because these two polynomials differ by scaling by a positive number $\lambda$.
    Therefore, $\lambda_0$ can be taken to be $\max(1, Cn^2)$.
\end{proof}

We will also use the following folklore analytic lemma.
\begin{lemma}\label{lem:conjugate}
    Let $\bx' = \Bf(\bx)$ a system of polynomial differential equations of dimension $n$ with an equilibrium point $\bx^\ast$.
    Let $\varphi \colon \mathbb{R}^n \to \mathbb{R}^n$ be an invertile change of coordinates, and let $\by' = \bg(\by)$ be the image of the system under the coordinate change. 
    Then the matrices $J_\bx(\Bf)|_{\bx = \bx^\ast}$ and $J_{\by}(\bg)|_{\by = \varphi(\bx^\ast)}$ are conjugate.
    In particular, they have the same eigenvalues.
\end{lemma}
\begin{proof}
    By the chain rule, we have
    \[
    \by' = (\varphi(\bx))' = J_{\bx}(\varphi) \bx' = J_{\bx}(\varphi)|_{\bx = \varphi^{-1}(\by)} f(\varphi^{-1}(\by)) = \bg(\by).
    \]
    Then we can write $J_\by(\bg)$ as
    \[
    J_{\bx}(\varphi)|_{\bx = \varphi^{-1}(\by)} J_{\bx}(f)|_{\bx = \varphi^{-1}(\by)} J_{\by}(\varphi^{-1}) + \sum\limits_{i = 1}^n A_i f_i(\varphi^{-1}(\by)),
    \]
    where $A_i$ is the Jacobian of the $i$-th column of $J_{\bx}(\varphi)|_{\bx = \varphi^{-1}(\by)}$.
    If we plug $\varphi(\bx^\ast)$ for $\by$, since $\Bf(\bx^\ast) = 0$, the latter sum will vanish, so we get
    \[
    J_{\bx}(\varphi)|_{\bx = \bx^\ast} J_{\bx}(f)|_{\bx = \bx^\ast} J_{\by}(\varphi^{-1})|_{\by = \varphi(\bx^{\ast})}.
    \]
    By the chain rule, the matrices $J_{\bx}(\varphi)|_{\bx = \bx^\ast}$ and $J_{\by}(\varphi^{-1})|_{\by = \varphi(\bx^{\ast})}$  are inverses to each other, so the Jacobians are indeed conjugated.
\end{proof}

\begin{proof}[Proof of Proposition~\ref{prop:gauging}]
    Before starting the proof, we would like to point at Example~\ref{eg:compute_proof2} in the Appendix illustrating the main steps of the proof.
    
    We define a map $\varphi\colon\mathbb{R}^{n + m} \to \mathbb{R}^{n + m}$ from a space with coordinates $(\bx, \by)$ to a space with coordinates $(\bx, \bz)$, where $\bz = (z_1, \ldots, z_m)$, by
    \begin{align*}
    &\varphi_i(\bx, \by) = x_i \quad \text{ for } 1 \leqslant i \leqslant n,\\
    &\varphi_{n + j}(\bx, \by) = y_j - g_j(\bx)\quad \text{ for } 1\leqslant j \leqslant m.
    \end{align*}
This map is invertible with the inverse given by $(\varphi^{-1})_i(\bx, \bz) = x_i$ for $1 \leqslant i \leqslant n$ and $(\varphi^{-1})_{n + j}(\bx, \bz) = z_j + g_j(\bx)$ for $1 \leqslant j \leqslant m$, so $\varphi$ is bijective.
Note that $\varphi(\bx^\circ, \bg(\bx^\circ)) = (\bx^\circ, \mathbf{0})$ for every $\bx^\circ \in \mathbb{R}^n$.
We apply a change of coordinates defined by $\varphi$ to the quadratic system~$\bx' = \bq_1(\bx, \by), \; \by' = \bq_2(\bx, \by)$ and obtain a (not necessarily quadratic) system of the form:
\begin{equation}\label{eq:transformed}
  \bx' = \widetilde{\bq}_1(\bx, \bz) \quad \text{ and } \quad \bz' = \widetilde{\bq}_2(\bx, \bz),
\end{equation}
where $\widetilde{q}_1 = q_1(\bx, \bz + \bg(\bz))$ and $\widetilde{q}_2$ can be found using the chain rule as follows:
\begin{equation}\label{eq:qtilde}
\bz' = (\by - \bg(\bx))' = q_2(\bx, \bz + \bg(\bx)) - J_{\bx}(\bg) \bx' = q_2(\bx, \bz + \bg(\bx)) - J_{\bx}(\bg) q_1(\bx, \bz + \bg(\bx)).
\end{equation}
Since the variety $\{(\bx^\circ, \bg(\bx^\circ)) \mid \bx^\circ \in \mathbb{R}^n\}$ was an invariant variety of~$\bx' = \bq_1(\bx, \by), \; \by' = \bq_2(\bx, \by)$ by construction, the linear space $\{(\bx^\circ, \mathbf{0}) \mid \bx^\circ\in \mathbb{R}^n \}$ is invariant for~\eqref{eq:transformed} and the restriction of~\eqref{eq:transformed} to this space coincides with the original system~$\bx' = \bp(\bx)$.
This implies the following constraints on~$\widetilde{\bq}_1$ and~$\widetilde{\bq}_2$:
\begin{itemize}
    \item $\widetilde{\bq}_1(\bx, \bz) = \bp(\bx) + \mathcal{O}(\bz)$, where $\mathcal{O}(\bz)$ stands for a polynomial with each monomial containing at least one of the $\bz$;
    \item $\widetilde{\bq}_2(\bx, \bz) = \mathcal{O}(\bz)$.
\end{itemize}
Due to these constraints, for every $\bx^{\circ} \in \mathbb{R}^n$, the Jacobian of $(\widetilde{\bq}_1, \widetilde{\bq}_2)$ at $(\bx^\circ, \mathbf{0})$ is of the form
\begin{equation}\label{eq:Jac_extended}
J_{\bx, \bz}(\widetilde{\bq}_1, \widetilde{\bq}_2)|_{\bx = \bx^{\circ}, \bz = \mathbf{0}} = \begin{pmatrix}
    J_{\bx}(\bp)|_{\bx = \bx^\circ} & \ast \\
    0 & J_{\bz}(\widetilde{\bq}_2)|_{\bx = \bx^\circ, \bz = \mathbf{0}}
\end{pmatrix}
\end{equation}
Let $h_1(\bx, \by), \ldots, h_m(\bx, \by)$ be the stabilizers of the quadratization (see Definition~\ref{def:stabilizer}).
We take an arbitrary parameter $\lambda \in \mathbb{R}$ and consider $\bq_{2, \lambda}(\bx, \by)$ defined~by
\begin{equation}\label{eq:q2lambda}
  \bq_{2, \lambda}(\bx, \by) = \bq_{2}(\bx, \by) - \lambda \bh(\bx, \by).
\end{equation}
Since the $h_i$'s are stabilizers, $\bg, \bq_{1}, \bq_{2, \lambda}$ is a quadratization of the original system for any value of $\lambda$ (see Definition~\ref{def:stabilizer}).
By using $q_{2, \lambda}$ instead of $q_2$ in~\eqref{eq:qtilde}, we
obtain $\widetilde{\bq}_{2, \lambda} = \widetilde{\bq}_2 - \lambda \bh(\bx, \bz + \bg(\bx))$.
Then, as in~\eqref{eq:Jac_extended}, we get
\begin{equation}\label{eq:Jac_lambda}
J_{\bx, \bz}(\widetilde{\bq}_1, \widetilde{\bq}_{2, \lambda})|_{\bx = \bx^{\circ}, \bz = \mathbf{0}} = \begin{pmatrix}
    J_{\bx}(\bp)|_{\bx = \bx^\circ} & \ast \\
    0 & (J_{\bz}(\widetilde{\bq}_2) - \lambda J_{\bz}(\bh))|_{\bx = \bx^\circ, \bz = \mathbf{0}}
\end{pmatrix}
\end{equation}
Observe that, since every $g_i$ is of the form $z_i$ plus polynomial in $\bx$ and $z$'s with smaller indices, 
$J_{\bz}(\bh)$ is a lower-triangular matrix with ones on the diagonal.

Having such a convenient expression for the Jacobian, we consider the given equilibria $\bx^\ast_1, \ldots, \bx^\ast_\ell$.
For any $\bx^\circ \in \{\bx_1^\ast, \ldots, \bx_\ell^\ast\}$, the eigenvalues of the Jacobian~\eqref{eq:Jac_lambda} are the union of the eigenvalues of $J_{\bx}(\bp)|_{\bx = \bx^\circ}$ and the eigenvalues of $(J_{\bz}(\widetilde{\bq}_2) - \lambda J_{\bz}(\bh))|_{\bx = \bx^\circ, \bz = \mathbf{0}}$.
Applying Lemma~\ref{lem:add_triang} to $\ell$ pairs of matrices $A = J_{\bz}(\widetilde{\bq}_2)|_{\bx = \bx_i^\ast, \bz = 0}$ and $B = J_{\bz}(\bh)|_{\bx = \bx_i^\ast, \bz = 0}$, we choose $\lambda$ to be larger than any of the $\lambda_0$'s provided by the lemma.
Then all the eigenvalues of this block will also have negative real parts.
By Lemma~\ref{lem:conjugate}, the same true for the Jacobian of $\bx' = \bq_1(\bx, \by),\; \by' = \bq_{2, \lambda}(\bx, \by)$.
\end{proof}

\section{Algorithms} \label{sec:alg}

Based on the proof of Theorem~\ref{thm:main}, finding a dissipativity-preserving quadratization can be done in two following steps:
\begin{enumerate}[label=\textbf{(Step \arabic*)}, leftmargin=16mm]
    \item finding an inner-quadratic quadratization
    \item modifying the corresponding quadratic system to achieve dissipativity at the given equilibria.
\end{enumerate}
In this section, we give algorithms for both steps.
Section~\ref{sec:inner} shows how to modify the quadratization algorithm from~\cite{bychkov2021optimal} to search for inner-quadratic quadratizations.
Using this algorithm as a building block, we give a general algorithm for computing dissipativity-preserving quadratizations in Section~\ref{sec:many_equilibria}.

\subsection{Computing inner-quadratic quadratization}\label{sec:inner}
Our algorithm follows the general  \emph{Branch-and-Bound (B\&B)} paradigm ~\cite{morrison2016branch} and is implemented based on the optimal monomial quadratization algorithm from~\cite[Section 4]{bychkov2021optimal}.
Therefore, we will describe the algorithm briefly, mainly focusing on the differences with the algorithm from~\cite{bychkov2021optimal}.

We define each subproblem~\cite[Definition 3.3]{bychkov2021optimal} as a set of new monomial variables $\{y_1(\bx), \ldots, y_{\ell}(\bx)\}$, and the subset of the \textit{search space}~\cite[Definition 3.1]{bychkov2021optimal} for the subproblem will be the set of all \textit{quadratizations} including these new variables. 
To each subproblem $\{y_1(\bx), \ldots, y_{\ell}(\bx)\}$, the algorithm from~\cite{bychkov2021optimal} assigns a set of generalized variables $\mathrm{V}$ (new variables, $\bx$'s, and $1$) and a set of nonsquares $\mathrm{NS}$ (monomials in the right-hand side which are not quadratic in the generalized variables~\cite[Definition 4]{bychkov2021optimal}). 
Additionally, we define the set of non-inner-quadratic new variables $\mathrm{NQ}$ which consist of all the monomials among $y_1(\bx), \ldots, y_{\ell}(\bx)$ which are not quadratic in $\{y_1(\bx), \ldots, y_{\ell}(\bx), x_1, \ldots, x_n\}$.
In particular, a subproblem is an inner-quadratic quadratization \textit{if and only if} $\mathrm{NS}=\varnothing$ and $\mathrm{NQ}=\varnothing$.
Note that $\mathrm{NS}$ and $\mathrm{NQ}$ are disjoint since $\mathrm{NQ} \subseteq \mathrm{V}$ and $\mathrm{V} \cap \mathrm{NS} = \varnothing$. 
\begin{example}
The notation previously introduced will now be demonstrated through the system $x'=x^{4}+x^{3}$ (taken from~\cite[Example 4]{bychkov2021optimal}, to display the difference between two algorithms).
We consider a subproblem with one already added new variable $y_{1}(x) := x^{3}$ (so, $y_{1}'=3x^{2}x'=3x^{6}+3x^{5}$). 
In this case, we have
$$
V=\{1, x, x^{3}\}, \quad V^{2}= \left\{1, x, x^2, x^3, x^4, x^6\right\}, \quad NS=\{x^{5}\}, \quad \underline{{NQ}=\{x^{3}\}}.
$$
\end{example}

\indent The algorithm starts from the subproblem $\varnothing$. For every iteration, we select one element $m$ from $\mathrm{NS} \cup \mathrm{NQ}$ (using a heuristic score function~\cite[Section 4.1]{bychkov2021optimal}) and compute all the decompositions of the form $m=m_{1}m_{2}$, where $m_1$ and $m_2$ are monomials. 
If $m \in \mathrm{NS}$, for every such decomposition, we create a new subproblem by adding the elements of $\left\{m_1, m_2\right\} \backslash V$ and at least one new variable will be added due to the property of $\mathrm{NS}$. 
If $m \in \mathrm{NQ}$, we only do this for the decompositions with $m_1 \neq 1$ and $m_2\neq 1$.

We apply this operation recursively and stop when $\mathrm{NS} \cup \mathrm{NQ}=\varnothing$ for each branch. 
Therefore, we can find all the possible inner-quadratic quadratization of the system. 
To improve the efficiency, we do not consider branches with more new variables than in already found answers and use versions of domain-specific pruning rules from~\cite{bychkov2021optimal}. 
The algorithm is summarized as Algorithm~\ref{alg:inner-quadratization}.
{\small
\begin{algorithm}[!h]
\caption{Computing optimal inner-quadratic quadratization }\label{alg:inner-quadratization}
\begin{description}
    \item[Input] 
    \begin{itemize}
        \item[]
        \item[-] polynomial ODEs system~$\bx' = \bp(\bx)$.
        \item[-] a set of already chosen new variables $y_1(\bx), \ldots, y_{\ell}(\bx)$ (at the first call, $\varnothing$).
        \item[-] the order $N$ of the smallest inner-quadratic quadratization found so far (at the first call, $N = \infty$).
    \end{itemize}
    \item[Output] a more optimal inner-quadratic  quadratization containing $y_1(\bx), \ldots, y_{\ell}(\bx)$ if such quadratization exists.
\end{description}

\begin{enumerate}[label = \textbf{(Step~\arabic*)}, leftmargin=*, align=left, labelsep=2pt, itemsep=0pt]
    \item If $y_1(\bx), \ldots, y_{\ell}(\bx)$ is a inner-quadratic quadratization, that is, $\mathrm{NS}=\varnothing$ and $\mathrm{NQ}=\varnothing$, and $\ell < N$, \textbf{return} $y_1, \ldots, y_\ell$.
    \item Select the element $m \in \mathrm{NS} \cup \mathrm{NQ}$ with the smallest \textit{score}, compute all the decompositions $m=m_{1}m_{2}$ as a product of two monomials. If $m \in NQ$, we only consider the decompositions with $m_{1} \neq 1$ and $m_{2} \neq 1$.
    \item For each decomposition $m = m_1m_2$ from the previous step, we consider a subproblem $\{z_1, \ldots, z_\ell\} \cup (\{m_1, m_2\} \setminus V)$.
    If its size is less than $N$ and none of the pruning rules apply, we run recursively on this subproblem and update $N$ if a more optimal inner-quadratic quadratization has been found by the recursive call.
\end{enumerate}
\end{algorithm}}

%%%%%%%%

%%%%%%%%

\subsection{Computing dissipativity-preserving quadratization}\label{sec:many_equilibria}

Based on the proof of Theorem~\ref{thm:main}, the main idea behind the search for dissipativity-preserving quadratization is to start with any inner-quadratic quadratization, and replace the right-hand side for the new variables, $\bq_2(\bx, \by)$, by $\bq_2(\bx, \by) - \lambda \bh(\bx, \by)$ (see~\eqref{eq:q2lambda}) for increasing values of $\lambda$ until the desired quadratization is found.
The detailed algorithm is given as Algorithm~\ref{alg:multi_dissipativity}, the proof of its correctness and termination is provided by Proposition~\ref{prop:algo_correctness}, and a step-by-step example is given in Example~\ref{ex:algo}.

{\small
\begin{algorithm}[!h]
\caption{Computing a quadratization dissipative at all provided equilibria}\label{alg:multi_dissipativity}
\begin{description}
    \item[Input] polynomial ODE system~$\bx' = \bp(\bx)$ and a list of its dissipative equilibria $\bx^\ast_1, \ldots, \bx_\ell^\ast$;
    \item[Output] a quadratization of the system which is dissipative at $\bx_1^\ast, \ldots, \bx_\ell^\ast$.
\end{description}

\begin{enumerate}[label = \textbf{(Step~\arabic*)}, leftmargin=*, align=left, labelsep=2pt, itemsep=0pt]
    \item\label{step:compute_quadr} Compute an inner-quadratic quadratization of~$\bx' = \bp(\bx)$ using \textbf{Algorithm~\ref{alg:inner-quadratization}}. Let the introduced variables be $y_1 = g_1(\bx), \ldots, y_m = g_m(\bx)$. Let $\bq_1(\bx, \by)$ and $\bq_2(\bx, \by)$ be the right-hand sides of the quadratic system as in Definition~\ref{def:quadr}.
    If the corresponding quadratic system is dissipative at $\bx_1^\ast, \ldots, \bx_\ell^\ast$, \textbf{return} it.
    \item\label{step:stabilizer} Construct the stabilizers $\bh(\bx, \by)$ for the quadratization from~\ref{step:compute_quadr} as in Definition~\ref{def:stabilizer}, and set $\lambda = 1$.
    \item\label{step:while} While \textbf{True}
    \begin{enumerate}
        \item\label{step:form_system} Construct a quadratic system $\Sigma_\lambda$
        \[
          \bx' = \bq_1(\bx, \by), \; \by' = \bq_2(\bx, \by) - \lambda \bh(\bx, \by).
        \]
        \item\label{step:check_return} Check if $\Sigma_\lambda$ is dissipative at $(\bx_i^\ast, \bg(\bx_i^\ast))$ for every $1 \leqslant i \leqslant \ell$ (using the Routh-Hurwitz criterion~\cite[Chapter~XV]{Gantmacher}).
        If yes, \textbf{return} quadratization defined by $\bg(\bx), \bq_1(\bx, \by), \bq_2(\bx, \by) - \lambda \bh(\bx, \by)$.
        Otherwise, set $\lambda = 2\lambda$.
    \end{enumerate}
\end{enumerate}
\end{algorithm}}

\begin{proposition}\label{prop:algo_correctness}
    Algorithm~\ref{alg:multi_dissipativity} always terminates and produces a correct output.
\end{proposition}

\begin{proof}
    We will start with proving the correctness. 
    Note that since $\bg, \bq_1, \bq_2$ computed in~\ref{step:compute_quadr} yield a quadratization of the input system, and $\bh$ vanishes if $\by$ is replaced with $\bg(\bx)$, then $\bg, \bq_1, \bq_2 - \lambda \bh$ yield a quadratization of the original system as well.
    Furthermore, if the algorithm returned at~\ref{step:check_return}, then this quadratization is dissipative at $\bx_1^\ast, \ldots, \bx_\ell^\ast$.

    The termination of the algorithm follows from the proof of Proposition~\ref{prop:gauging}.
    We observe that the constructed $\bq_2 - \lambda \bh$ is exactly $\bq_{2, \lambda}$ in the notation of the proof, and it is shown that there exists $\lambda_0$ such that, for every $\lambda > \lambda_0$, $\bg, \bq_1, \bq_{2, \lambda}$ is dissipative at $\bx_1^\ast, \ldots, \bx_\ell^\ast$.
    Since $\lambda$ in the algorithm is doubled on each iteration of the while loop, it will at some point exceed $\lambda_0$, and the algorithm will terminate.
\end{proof}

\begin{example}\label{ex:algo}
We will illustrate how Algorithm \ref{alg:multi_dissipativity} works with the following differential equation:
\begin{equation}\label{eq:large_lambda}
    x' = -x (x - a) (x - 2a)
\end{equation}
where $a$ is a positive scalar parameter. 
The system's equilibria are $0, a, 2a$, and, among them, $x=0$ and $x=2a$ are dissipative.
Regardless of the value of $a$, Algorithm~\ref{alg:inner-quadratization} called at~\ref{step:compute_quadr} will produce an inner-quadratic quadratization with one new variable $y = x^2$ and quadratic system:
\[
\begin{cases}
    x' = -xy + 3ax^{2} -2a^2x, \\
    y' = -2y^{2}+6axy-4a^2x^{2}
\end{cases}
\]
The stabilizer computed at~\ref{step:stabilizer} will be $h(x, y) = y - x^2$.

Now let us fix $a = 1$ and continue with~\ref{step:while}.
At~\ref{step:form_system}, we form a new quadratic system $\Sigma_\lambda$:
\[
\begin{cases}
    x' = -xy + 3x^{2} -2x, \\
    y' = -2y^{2}+6xy-4x^{2} - \lambda (y-x^{2})
\end{cases}
\]
For $\lambda = 1, 2, 4, 8, \ldots$ we check the eigenvalues of its Jacobian at points $(0, 0)$ and $(2, 4)$. 
The Jacobian of $\Sigma_\lambda$ is
\[
J=\left[\begin{array}{cc}
-y+6 x-2 & -x \\
6 y+2 \lambda x-8 x\;\; & -4 y-\lambda+6 x
\end{array}\right]
\]
The eigenvalues we get on different iterations of the while-loop are summarized in Table~\ref{tab:eigenvalues} (the ones with nonnegative real parts are bold).
\begin{table}[htbp!]
    \renewcommand{\arraystretch}{1.2} % Adjust line spacing
\setlength{\tabcolsep}{8pt} % Adjust column spacing
    \centering
    \begin{tabular}{|c|c|c|}
    \hline
       $\lambda$ & at $(0, 0)$ & at $(2, 4)$ \\ \hline\hline
     $1$ & $-2, \; -1$ & $-2,\; \mathbf{3}$ \\ \hline
     $2$ & $-2,\; -2$ & $-2, \; \mathbf{2}$ \\ \hline
     $4$ & $-2,\; -4$ & $-2, \; \mathbf{0}$ \\ \hline
     $8$ & $-2,\; -8$ & $-2,\; -4$ \\
     \hline
    \end{tabular}
    \caption{Eigenvalues of the Jacobian of $\Sigma_\lambda$ at~\ref{step:check_return}}\label{tab:eigenvalues}
\end{table}

From the table we see, that the algorithm will stop and return at $\lambda = 8$.
Note that our implementation offers three way of verifying the dissipativity: by computing the eigenvalues directy numerically (with {\sc numpy}) or symbolically (with {\sc sympy}) or by using the Routh-Hurwitz criterion~\cite[Chapter~XV]{Gantmacher}, \cite[Chapter~3]{RHbook} (via {\sc tbcontrol} package~\cite{tbcontrol}).
The numerical evaluation of eigenvalues is the fastest (see Tables~\ref{tab:performance_lambda} and~\ref{tab:performance}) but does not yield fully rigorous guarantees, the other two methods may be slower but provide such guarantees.

As the value of $a$ increases, the original system is more unstable at equilibrium $x_{eq} = 2a$, which requires a larger value of $\lambda$ in order to make the system dissipative at $(x_{eq}, x_{eq}^{2})$. We compute the dissipative quadratization of the system~\eqref{eq:large_lambda} with different values of $a$ and the running time for each method, which is presented in Table \ref{tab:performance_lambda}.
\begin{table}[!h]
\centering
\renewcommand{\arraystretch}{1.2} % Adjust line spacing
\setlength{\tabcolsep}{8pt} % Adjust column spacing
\begin{tabular}{|c|c|c|c|c|}
\hline
$a$ & $\lambda$ (output) & time (\textit{Numpy}) & time (\textit{Routh-Hurwitz}) & time (\textit{Sympy}) \\
\hline
\hline
1 & 8 & 33.63 & 36.89 & 40.98 \\ \hline
5 & 128 & 34.04 & 38.36 & 39.84 \\ \hline
10 & 512 & 33.07 & 41.91 & 43.66 \\ \hline
50 & 16384 & 33.38 & 41.18 & 54.70 \\ \hline
100 & 65536 & 34.90 & 43.06 & 54.31 \\ \hline
\end{tabular}\newline
\caption{Output $\lambda$ value and runtimes (in milliseconds) with different methods for $a$ of the system \ref{eq:large_lambda}, results were obtained on a laptop Apple M2 Pro CPU @ 3.2 GHz, MacOS Ventura 13.3.1, CPython 3.9.1. 
Runtime is averaged over $10$ executions.}
\label{tab:performance_lambda}
\end{table}
\end{example}

%%%%%%%%%%%%%%%%%%%%%%%%%%%%%%%%%%%%

\section{Case studies}\label{sec:case_studies}

The code for reproducing the results of the case studies below is available in the ``Examples'' folder of~\cite{code}.

\subsection{Application to reachability analysis}\label{sec:reachability}

The reachability problem is: given an ODE system $\bx' = \bp(\bx)$, a set $S\subseteq \mathbb{R}^n$ of possible initial conditions, and a time $t \in \mathbb{R}_{>0}$, compute a set containing the set
\[
\{\bx(t) \mid \bx'=\bp(\bx) \;\&\; \bx(0) \in S\} \subseteq \mathbb{R}^n
\]
of all points reachable from $S$ at time $t$.
One recent approach to this problem in the vicinity of a dissipative equilibrium $\bx^\ast$ proposed by Forets and Schilling in~\cite{Forets_2021} is to use Carleman linearization to reduce the problem to the linear case which is well-studied.
However, the approach described in~\cite{Forets_2021} relied on explicit bounds available
only for quadratic systems under the assumption of dissipativity and weak nonlinearity (see~\cite[definition 1 and 2]{Forets_2021}).
Algorithm~\ref{alg:multi_dissipativity} allows this restriction to be relaxed by computing a quadratization which preserves the dissipativity of $\bx^\ast$.

We will illustrate this idea using the \emph{Duffing} equation
\[
  x'' = kx + ax^{3} + bx'
\]
which describes a damped oscillator with non-linear restoring force.
The equation can be written as a first-order system by introducing $x_1 := x, x_2 := x'$ as follows
\[
    x_1' = x_2,\quad 
    x_2' = kx_1 + a x_1^3 + b x_2.
\]
We take $a = 1, b = -1, k = 1$. 
Then the system will have three equilibria $\mathbf{x}^\ast = (0, 0), (-1, 0), (1, 0)$, among which it will be dissipative only at the origin.
Algorithm~\ref{alg:inner-quadratization} finds an inner-quadratic quadratization for the system using a new variable $y(\bx)=x_1^2$ resulting in the following quadratic system:
$$
x_1^{\prime}=x_2,\quad
x_2^{\prime}=a x_1 y +b x_2+k x_1,\quad
y^{\prime}=2 x_1 x_2.
$$
Obviously, the quadratization is an inner-quadratic quadratization as well. 
By applying Algorithm \ref{alg:multi_dissipativity}, we get $\lambda = 1$ with the following dissipative quadratization:
\begin{equation} \label{eq:reachability_system}
\left\{\begin{array}{l}
x_1^{\prime}=x_2 \\
x_2^{\prime}= x_1y + x_1 - x_2 \\
y^{\prime}=-y+x_1^2+2 x_1 x_2
\end{array}\right.
\end{equation}

For the initial conditions $x_1(0) = 0.1, \; x_2(0) = 0.1, \; y(0) = x_1(0)^2 = 0.01$, system~\eqref{eq:reachability_system} satisfies the requirement of the algorithm from~\cite{Forets_2021}.
We apply the algorithm with truncation order $N = 5$ and report the result of the reachability analysis in Figure \ref{fig:reachability}.
The grey curve is the computed trajectory and the blue area is an upper bound for the reachable set.
\begin{figure}[h!] 
    \centering
    \includegraphics[width=0.9\textwidth]{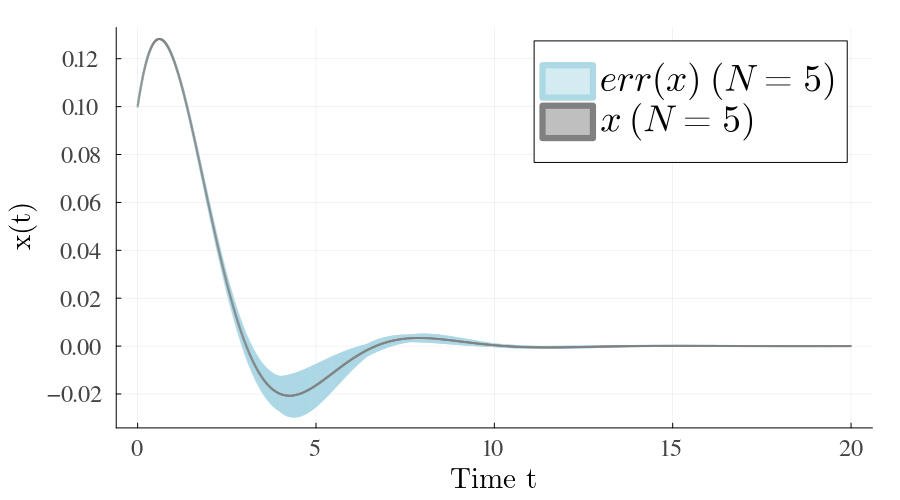}
    \caption{Reachability analysis results with the computed trajectory (gray) and overapproximation of the reachable set (light blue).
    Initial condition $\mathcal{X}_0 = [0.1, 0.1, 0.01]$, truncation order $N = 5$, and the estimate reevaluation time $t = 4$ (see~\cite[Section~6.1]{Forets_2021}).}
    \label{fig:reachability}
\end{figure}

\subsection{Preserving bistability}

An ODE model is called \emph{bistable} (or multistable) if it has at least two stable equilibria.
This is a fundamental property for models in life sciences since such a model describes a system that can exhibit a switch-like behaviour, in other words, ``make a choice''~\cite{craciun}.
One of the smallest possible bistable models arising from a simple chemical reaction network~\cite[Table 1]{wilhelm} is given by the following scalar ODE:
\[
x' = k_1 x^2 - k_2 x^3 - k_3 x,
\]
where $k_1, k_2, k_3$ are positive reaction rate constants.
The equation has always one dissipative equilibrium at $x = 0$.
It has two more equilibria as long as $k_1^2 > 4k_2k_3$, and in this case, the largest of them will be dissipative as well.
For any nonzero parameter values, the inner-quadratic quadratization computed by Algorithm~\ref{alg:inner-quadratization} will consist of a single new variable $y(x) := x^2$ and the quadratic system:
\begin{equation}\label{eq:multistable_lifted}
  x' = k_1w - k_2xw - k_3x, \quad y' = 2k_1xy - 2k_2y^2 - 2k_3y.
\end{equation}
For the case-study, we pick $k_1 = 0.4, k_2 = 1, k_3 = 0.03$.
For these parameter values, the dissipative equilibria are $x = 0$ and $x = 0.3$, and Algorithm~\ref{alg:multi_dissipativity} finds that~\eqref{eq:multistable_lifted} is dissipative at them already. The plot below shows that, indeed, the trajectories of~\eqref{eq:multistable_lifted} staring in the neighbouthoods of $(0, 0)$ and $(0.3, 0.09)$ converge to the respective equilibria.
\begin{figure}[!htbp]
    \centering
    \includegraphics[width=0.9\textwidth]{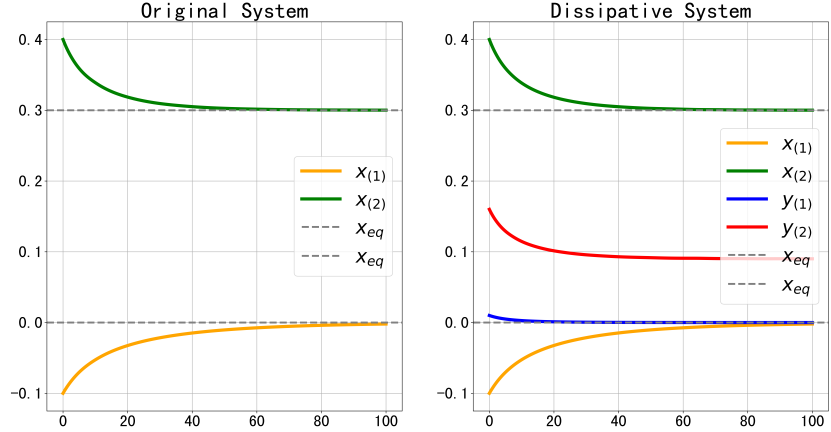}
    \caption{Plot of the original equation and system (\ref{eq:multistable_lifted}) with initial state $\mathcal{X}_{0}=[x_{0},w_{0}]=[-0.1,0.01]$ ($x_{(1)}$, $w_{(1)}$) and $\mathcal{X}_{0}=[0.4,0.16]$ ($x_{(2)}$, $w_{(2)}$).}
    \label{fig:bistable}
\end{figure}

\subsection{Coupled Duffing oscillators}

For a larger example, we will consider an ensemble consisting of Duffing oscillators from Section~\ref{sec:reachability} which is an extended version of a pair of coupled oscillators from~\cite{Balamurali2022}.
The model consisting of $n$ oscillators is parametrized by a number $\delta \in \mathbb{R}$ and a matrix $A \in \mathbb{R}^{n \times n}$, and is defined by the following system:
\[
\bx'' = A\bx - (A\bx)^3 - \delta \bx', 
\]
where $\bx = [x_1, \ldots, x_n]^\top$ are the positions of the oscillators, and $(A\bx)^3$ is the component-wise cube of vector $A\bx$.
Similarly to Section~\ref{sec:reachability}, we can rewrite this as a first-order system of dimension $2n$ by introducing new variables $\bz = [z_1, \ldots, z_n]^\top$ for the derivatives of $\bx$:
\[
\dot{\bx} = \bz, \quad \bz' = A\bx - (A\bx)^3 - \delta \bz.
\]
Similarly, to~\cite[Table~1]{Balamurali2022}, if the eigenvalues of $A$ are positive real numbers, then this system has $2^n$ dissipative equilibria.
We run our code for $n = 1, \ldots, 8$ taking $\delta = 2$ and $A$ being the tridiagonal matrix with ones on the diagonal and $\frac{1}{3}$ on the adjacent diagonals.
Table~\ref{tab:performance} reports, for each $n$, the number of introduced variables and the times for computing inner-quadratic quadratization (Algorithm~\ref{alg:inner-quadratization}) and making it dissipative at all $2^n$ equilibria (Algorithm~\ref{alg:multi_dissipativity} using {\sc numpy} for the eigenvalue computation or the symbolics Routh-Hurwitz criterion).
We can observe that numerical methods for checking the dissipativity scale well (given that the number of points grows exponentially) while symbolic methods become very costly as the dimension grows.
\begin{table}[!hbtp]
\centering
\renewcommand{\arraystretch}{1.2} % Adjust line spacing
\begin{tabular}{|c|c|c|c|c|c|c|}
\hline
\multirow{2}{*}{$n$} & \multirow{2}{*}{dimension} & \multirow{2}{*}{\# equilibria} & \multirow{2}{*}{\# new vars} & \multirow{2}{*}{time (inner-quadratic)} & \multicolumn{2}{c|}{time (dissipative)}    \\ \cline{6-7} 
                     &                            &                                &                              &                                         & {\sc numpy} & {\sc Routh-Hurwitz} \\ \hline
1                    & 2                          & 2                              & 1                            & 0.02                                    & 0.05  &     0.07          \\ \hline
2                    & 4                          & 4                              & 2                            & 0.07                                    & 0.19  &     0.65          \\ \hline
3                    & 6                          & 8                              & 4                            & 0.20                                    & 0.74  &     36.57          \\ \hline
4                    & 8                          & 16                             & 5                            & 0.39                                    & 1.62  &    1179.33           \\ \hline
5                    & 10                         & 32                             & 7                            & 0.72                                    & 4.30  &   > 2000            \\ \hline
6                    & 12                         & 64                             & 9                            & 1.20                                    & 11.28 &   > 2000            \\ \hline
7                    & 14                         & 128                            & 10                           & 1.75                                    & 28.23 &   > 2000            \\ \hline
8                    & 16                         & 256                            & 12                           & 2.63                                    & 78.70 &  > 2000             \\ \hline
\end{tabular}
\caption{Runtimes (in seconds) for $n$ coupled Duffing oscillators, results were obtained on a laptop with the following parameters: Apple M2 Pro CPU @ 3.2 GHz, MacOS Ventura 13.3.1, CPython 3.9.1.}
\label{tab:performance}
\end{table}

\vspace{-3mm}
\section{Conclusions}\label{sec:conclusions}

While various quadratization techniques have been used recently in a number of application areas, and in most of the cases this was primarily involving numerical simulations, we are not aware of prior general results on the stability properties of the quadratized systems.
In this paper, we studied quadratizations that preserve dissipativity at prescribed equilibria. 
First, we have shown that, for any set of dissipative equilibria such a quadratization exists.
Then we have presented an algorithm capable of computing a quadratization with this property with dimension low enough to be of interest for applications.
We showcase the algorithm on several case studies, including examples from reachability analysis and chemical reaction network theory.

The key ingredient of our algorithm is the computation of a quadratization (we call it inner-quadratic) which gives us substantial control over the stability properties of the quadratized system.
We expect that this construction will be useful for further research in this direction.

In future research, we plan to extend the results of the paper in different directions.
One natural problem is to extend the results and algorithms from the present paper beyond polynomial systems, for example, by designing an algorithm for dissipativity-preserving polynomialization.
Additionally, exploring the preservation of other stability properties, such as limit cycles, attractors, and Lyapunov functions, is another promising avenue for research.

%%%%%%%%%%%%%%%%%%%%%%%%%%%%%%%%%%%%

\bibliographystyle{splncs04}
\bibliography{bibliography}

\appendix
\section{Example illustrating the proof of Proposition~\ref{prop:gauging}}

\begin{example}\label{eg:compute_proof2}
We use the \textbf{Example \ref{eg:stabilizer_example} (Stabilizers)} to illustrate the proof for Proposition \ref{prop:gauging}. Consider the following system:
$$
x_1^{\prime}=-3 x_1+x_2^4 , \quad 
x_2^{\prime}=-2 x_2+x_1^2.
$$
By applying \textbf{Algorithm \ref{alg:inner-quadratization}}, we introduce the following new variables to obtain the inner-quadratic system
$$
y_1=x_1^2, \quad
y_2=x_2^2, \quad
y_3=x_1 x_2, \quad
y_4=x_2^3 = x_{2}y_{2}.
$$
Then we got the inner-quadratic system as follows:
$$
\bx'=\bq_{1}(\bx,\by): 
\begin{cases}
    x_1' &= -3x_1 + y_2^2 \\
x_2' &= -2x_2 + x_1^2
\end{cases}
\quad 
\by'=\bq_{2}(\bx,\by):
\begin{cases}
    y_1' &= 2y_4y_3 - 6x_1^2 \\
y_2' &= 2y_1x_2 - 4x_2^2 \\
y_3' &= y_4y_2 + y_1x_1 - 5x_1x_2 \\
y_4' &= 3y_1y_2 - 6y_2x_2
\end{cases}
$$
We apply the change of coordinates defined by $\varphi$ to the inner-quadratic system~$\bx' = \bq_1(\bx, \by), \; \by' = \bq_2(\bx, \by)$ and obtain the following system:
$$
\bx'=\widetilde{\bq}_1(\bx, \bz): 
\begin{cases}
    x_1' &= -3 x_1+x_2^4+2 x_2^2 z_2+z_2^2 \\
x_2' &= x_1^2-2 x_2
\end{cases}
$$
$$
\bz' = \widetilde{\bq}_2(\bx, \bz):
\begin{cases}
    z_1' &= -4 x_1 x_2^2 z_2+2 x_1 x_2 z_4-2 x_1 z_2^2+2 x_2^3 z_3+2 z_3 z_4 \\
z_2' &= 2 x_2 z_1 \\
z_3' &= x_1 z_1-x_2^3 z_2+x_2^2 z_4-x_2 z_2^2+z_2 z_4 \\
z_4' &= 3x_1^2 z_2+3 x_2^2 z_1-6 x_2 z_2+3 z_1 z_2
\end{cases}
$$
One can observe that the system satisfies to the constraints on~$\widetilde{\bq}_1$ and~$\widetilde{\bq}_2$ that $\widetilde{\bq}_1(\bx, \bz) = \bp(\bx) + \mathcal{O}(\bz)$ and $\widetilde{\bq}_2(\bx, \bz) = \mathcal{O}(\bz)$ where $\mathcal{O}(\bz)$ stands for a polynomial with each monomial containing at least one of the $\bz$. Then, we can compute the Jacobian matrix and obtain the following matrix:
\begin{equation}\label{eq:Jac_extended}
J_{\bx, \bz}(\widetilde{\bq}_1, \widetilde{\bq}_2)|_{\bx = \bx^{\circ}, \bz = \mathbf{0}} = \left[\begin{array}{cccccc}
-3 & 4 x_{2}^3 & 0 & 2 x_{2}^2 & 0 & 0 \\
2 x_{1} & -2 & 0 & 0 & 0 & 0 \\
0 & 0 & 0 & -4 x_{1} x_{2}^2 & 2 x_{2}^3 & 2 x_{1} x_{2} \\
0 & 0 & 2 x_{2} & 0 & 0 & 0 \\
0 & 0 & x_{1} & -x_{2}^3 & 0 & x_{2}^2 \\
0 & 0 & 3 x_{2}^2 & 3x_{1}^2-6 x_{2} & 0 & 0
\end{array}\right]
\end{equation}
We have the stabilizers of the quadratization as follows:
\begin{align*}
    h_{1}(\bx, \by) &=y_1-x_1^2, & h_{2}(\bx, \by) &=y_2-x_2^2,\\
    h_{3}(\bx, \by) &=y_3-x_1 x_2, &
    h_{4}(\bx, \by) &= y_4-x_{2}y_{2}
\end{align*}
We take an arbitrary parameter $\lambda \in \mathbb{R}$ and consider $\bq_{2, \lambda}(\bx, \by)$ as the definition in equation \ref{eq:q2lambda}:
\begin{equation}
  \bq_{2, \lambda}(\bx, \by) = \bq_{2}(\bx, \by) - \lambda \bh(\bx, \by)=
\begin{cases}
2y_4y_3 - 6x_1^2 - \lambda (y_1-x_1^2)\\
2y_1x_2 - 4x_2^2 - \lambda (y_2-x_2^2)\\
y_4y_2 + y_1x_1 - 5x_1x_2 - \lambda (y_3-x_1 x_2)\\
3y_1y_2 - 6y_2x_2 - \lambda (y_4-x_{2}y_{2})
\end{cases}
\end{equation}
With the same method of coordinates change, we obtain:
$$
\widetilde{\bq}_{2, \lambda} = 
\begin{cases}
& (-4 x_1 x_2^2 z_2+2 x_1 x_2 z_4-2 x_1 z_2^2+2 x_2^3 z_3+2 z_3 z_4) + -\lambda z_1 \\
& (2 x_2 z_1) -\lambda z_2 \\
& (x_1 z_1-x_2^3 z_2+x_2^2 z_4-x_2 z_2^2+z_2 z_4) -\lambda z_3\\
& (3x_1^2 z_2+3 x_2^2 z_1-6 x_2 z_2+3 z_1 z_2) - \lambda (z_{4}-x_{2}z_{2})
\end{cases}
= \widetilde{\bq}_2 - \lambda \bh(\bx, \bz + \bg(\bx))
$$
where we have
$$
\bh(\bx, \bz + \bg(\bx)) = 
\begin{cases}
    (z_{1}+x_{1}^{2})-x_{1}^{2} = z_{1}\\
    (z_{2}+x_{2}^{2})-x_{2}^{2}= z_{2}\\
    (z_{3}+x_{1}x_{2})-x_{1}x_{2}= z_{3}\\
    (z_{4}+x_{2}^{3})-x_{2}(z_{2}+x_{2}^{2}) = z_{4} - x_{2}z_{2}
\end{cases}
$$
Taking the Jacobian of $\bh$ against $\bz$, we can obtain the following matrix:
$$
\lambda J_{\bz}(\bh))|_{\bx = \bx^\circ, \bz = \mathbf{0}} = \lambda
\begin{bmatrix}
    1 & 0 & 0 & 0 \\
    0 & 1 & 0 & 0 \\
    0 & 0 & 1 & 0 \\
    0 & x_{2} & 0 & 1
\end{bmatrix}
$$
which is a lower-triangular matrix with ones on the diagonal.
\end{example}
\end{document}